\documentclass{amsart}

\usepackage[T1]{fontenc}
\usepackage{mathrsfs}
\usepackage{subcaption}
\usepackage{graphicx}
\usepackage{float}
\usepackage[sorting=none]{biblatex}
\usepackage[english]{babel}
\usepackage{csquotes}
\usepackage{url}
\usepackage{booktabs} 
\usepackage{microtype}
\usepackage{algpseudocode}
\usepackage[section]{algorithm}

\newtheorem{theorem}[subsection]{Theorem}
\newtheorem{lemma}[subsection]{Lemma}
\newtheorem{corollary}[subsection]{Corollary}
\newtheorem{conjecture}[subsection]{Conjecture}

\theoremstyle{definition}
\newtheorem{definition}[subsection]{Definition}
\newtheorem{example}[subsection]{Example}

\title{Stable Homology-Based Cycle Centrality Measures}
\author{John Rick D. Manzanares$^{\ast}$}
\address{Department of Mathematics and Computer Science, University of the Philippines Baguio, Baguio City, Philippines 2600}
\email{jdmanzanares@up.edu.ph}
\thanks{$^{\ast}$Corresponding author}

\author{Paul Samuel P. Ignacio}
\address{Department of Mathematics and Computer Science, University of the Philippines Baguio, Baguio City, Philippines 2600}
\email{ppignacio@up.edu.ph}

\date{\today}

\begin{document}

\begin{abstract}
Network centrality measures play a crucial role in understanding graph structures, assessing the importance of nodes, paths, or cycles based on directed or reciprocal interactions encoded by vertices and edges. Estrada and Ross extended these measures to simplicial complexes to account for higher-order connections. In this work, we introduce novel centrality measures by leveraging algebraically-computable topological signatures of cycles and their homological persistence. We apply tools from algebraic topology to extract multiscale signatures within cycle spaces of weighted graphs, tracking homology generators persisting across a weight-induced filtration of simplicial complexes built over point clouds. This approach incorporates persistent signatures and merge information of homology classes along the filtration, quantifying cycle importance not only by geometric and topological significance but also by homological influence on other cycles. We demonstrate the stability of these measures under small perturbations using an appropriate metric to ensure robustness in practical applications. Finally, we apply these measures to fractal-like point clouds, revealing their capability to detect information consistent with, and possibly overlooked by, common topological summaries.
\end{abstract}

\maketitle

\section{Introduction}\label{sec:intro}

Many complex networks, such as social networks \cite{socnet} and telecommunication networks \cite{telnet} use graph-based centrality measures to determine the relative significance of nodes or cycles in the network. The derivations of the measures of central tendency in Statistics reflect the idea that a single value can represent the entire distribution of a data set. In particular, the mode of data sets is comparable to the degree centrality of graphs in that it projects importance through frequencies. Similarly, closeness centrality is akin to the median of data in the sense that it identifies nodes that are reachable via short paths from any node as "central" nodes.

Giscard and Wilson \cite{loopcentrality} introduced the loop-centrality measure that uses the number of walks that intersect a loop to measure its importance. They found that this centrality measure has the ability to detect essential proteins in protein-protein interaction networks. For the same purpose, Estrada and Ross \cite{centrality} explored an extension of this centrality measure to finite-dimensional substructures where higher-order connectivity networks (such as co-author \cite{coauthor} and social contact \cite{soccont} networks) are represented by triangles and tetrahedra. A widely accepted notion of importance for cycles in a simplicial complex regards long-lived cycles as essential features of data, while short-lived cycles that appear are likely to be due to noise or sampling errors. However, Bubenik et al. \cite{curvature} demonstrated that short-lived cycles hold important information that can be used to estimate the curvature of surfaces.

In this study, we propose novel centrality measures that leverage the persistence of homology classes and their merge history along the filtration. Integral to this is the development of an algorithm that captures the merge history of homology classes. These homology-based centrality measures produce, for all cycle generators, curves that reflect the relative importance of the corresponding generator throughout its entire evolution. By applying these centrality measures on toy models, we demonstrate the consistency of detected information by these measures to other topological summaries, and highlight its ability to capture new information possibly missed by other summaries. Accordingly, we establish some properties that include the stability of these measures under a distance analogous to norms in Lebesgue spaces and persistence landscapes.

\section{Preliminaries}\label{sec:prelims}

This section lays the groundwork for extracting multiscale topological signatures from weighted graphs. First, we model higher-order interactions among vertices in the graph using simplicial complexes, similar to the approach used in loop centrality \cite{centrality}. This captures interactions beyond simple pairwise connections. Next, we refine this model by applying a weight-based filtration. Filtration progressively removes less important edges, resulting in a nested sequence of increasingly coarser simplicial complexes. Finally, we leverage the concept of simplicial homology to analyze these filtered complexes. Simplicial homology allows us to identify and track cycle generators, which are loops within the complex that cannot be continuously deformed into a point. By focusing on these generators at each stage of the filtration, we can build a multi-scale record of the graph's overall topological structure. This is crucial to designing our cycle centrality measures later. Throughout this chapter, we shall use \cite{ctop} and \cite{ten} as standard references for discussions involving simplicial and persistent homology.

\subsection{Simplicial Complex}\label{subsec:simpcomp}

We dive into the concept of simplicial homology and its application in modeling graphs. Although graphs represent connections between pairs of nodes (edges), simplicial complexes offer a more versatile framework. They allow us to capture not only pairwise interactions but also higher-order relationships between multiple nodes using simplices of various dimensions. This ability to encode complex interactions makes them particularly suitable for analyzing graphs with intricate connections beyond simple edges.
	
\begin{definition}\label{def:asc}
An \textit{(abstract) simplicial complex} is a collection $\mathscr{C}$ of subsets of a finite set $V$ such that $\tau \subseteq \sigma \in \mathscr{C}$ implies $\tau \in \mathscr{C}$. An element $\sigma\in \mathscr{C}$ is called an \textit{(abstract) simplex} with \emph{dimension} $|\sigma|-1$, and the largest such dimension among all simplices in $\mathscr{C}$ defines the dimension of $\mathscr{C}$. 
\end{definition} 	

In the Euclidean space, lower-dimensional simplices are named as \textit{vertex}, \textit{edge}, \textit{triangle}, and \textit{tetrahedron} for a 0-simplex, a 1-simplex, a 2-simplex, and a 3-simplex, respectively. Furthermore, the higher-dimensional simplices are polytopes analogous to triangles and tetrahedra.

We can construct a simplicial complex based on a metric or a weight function. A metric-based example is the \textit{Vietoris-Rips complex} or the \textit{Rips complex}.

\begin{definition}\label{def:rips}
Let $\epsilon > 0$. Suppose that $(M, d)$ is a metric space and $P = \{p_i\}$ is a finite subset of $M$. The \textit{Vietoris-Rips complex} $\mathscr{R}_\epsilon(P)$ of $P$, with threshold $\epsilon$, is a simplicial complex whose $k$-simplices correspond to $(k+1)$-tuples of points in $P$ such that $d(p_i,p_j) \leq 2\epsilon$ for any pair of integers $(i ,j)$.
\end{definition}

Rips complexes form an $n$-simplex when every pair from the $n+1$ points is connected by a $1$-simplex or an edge. Figure \ref{fig:rips} illustrates a Rips complex of a point cloud with six initial points from the Euclidean space.

\begin{figure}[ht]
\centering
\includegraphics[width = 0.45\linewidth]{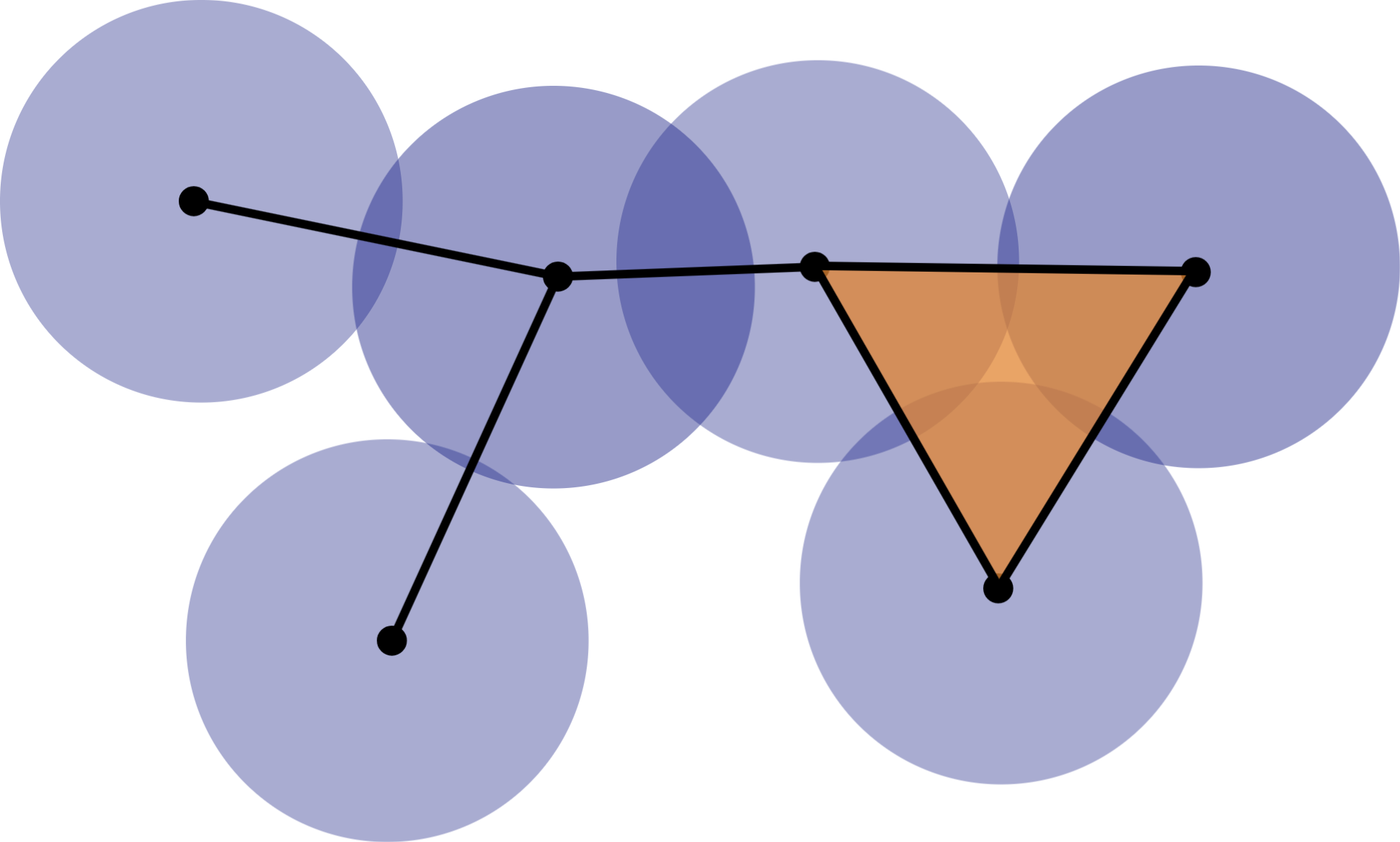}    
\caption{Rips complex of a point cloud}
\label{fig:rips}
\end{figure}

The simplices of an abstract simplicial complex provide the fundamental building blocks for constructing chains in chain space. Each $k$-simplex acts as a single unit within a $k$-chain. By combining these simplices with coefficients from $\mathbb{Z}/2\mathbb{Z}$ (which essentially encode the presence or absence of edges), we can create formal sums that represent more intricate relationships within the complex.

\begin{definition}\label{def:chains}
A \textit{$k$-chain} is a formal sum of $k$-simplices with coefficients coming from the field $\mathbb{Z}/2\mathbb{Z}$. The \emph{chain space}, denoted by $C_k$, is the free Abelian group generated by all possible $k$-simplices in the complex.
\end{definition}

These $k$-chains serve as the building blocks for studying the topological properties of the simplicial complex. However, not all $k$-chains represent actual cycles within the complex. To distinguish cycles from other chains, we introduce the concept of a boundary operator.

\begin{definition}\label{def:boundary_op}
    The \textit{boundary operator} $\partial_k : C_k \rightarrow C_{k-1}$ is the linear extension defined on the generators (individual $k$-simplices) given by
$$
\partial_k\left([x_0, \dots, x_j, \dots, x_n]\right) = \sum_{i = 0}^n[x_0,\dots,\hat{x_i},\dots,x_n]
$$
where $\hat{x_i}$ indicates that $x_i$ is omitted. 
\end{definition}

Mathematically, the boundary operator acts on a $k$-simplex by summing over all its $(k-1)$-dimensional faces. For each face, it creates a new $(k-1)$-dimensional chain with the opposite coefficient. Intuitively, this captures the idea that the boundary of a $k$-dimensional object is made up of its $(k-1)$-dimensional faces.

\begin{definition}\label{def:cycle}
The \emph{cycle space} $Z_k$ and \emph{boundary space} $B_k$ of $C_k$ are defined by
$$
Z_k = \ker \partial_k, \mbox{ and } B_k = Im\  \partial_{k+1}.
$$
We refer to elements of $Z_k$ as $k$-cycles and those of $B_k$ as $k$-boundaries.
\end{definition}

Intuitively, these chains in the cycle space represent closed loops ($k$-cycles) within the complex that cannot be continuously deformed into a single point. Meanwhile, the chains in the boundary space represent the edges or borders of higher-dimensional simplices.

Applying the boundary operator twice, we confirm that $\partial_k\circ \partial_{k+1} = 0$. This property implies that every element of the boundary space is also a cycle. 

\subsection{Simplicial Homology}\label{subsec:simphom}

Now that we can distinguish cycles and boundaries, we can introduce the concept of homology, which focuses on the essence of cycles, capturing their topological properties rather than their specific form.

\begin{definition}\label{def:homstuff}
Let $\mathscr{C}$ be a simplicial complex. The \textit{$k$th homology group} of $\mathscr{C}$, written as $H_k(\mathscr{C})$, is the quotient group $Z_k/B_k$. Two cycles in $Z_k(\mathscr{C})$ lying in the same homology class in this space are said to be \textit{homologous}. We refer to the rank of $H_k(\mathscr{C})$ as the \textit{$k$th Betti number} of $\mathscr{C}$.  
\end{definition}		

In other words, the homology group is formed by taking the cycle space and "glue" any cycles that differ only by a boundary. This captures the intrinsic topological features of the complex at dimension $k$, independent of the specific choices of representatives for these cycles. The Betti number provides a numerical measure of the number of independent holes of dimension $k$ in the complex.

By modding out the cycle space by the boundary space, we effectively exclude cycles that are simply the boundaries of higher-dimensional simplices. This ensures that we focus on nontrivial cycles that capture the true topological features of the complex. Lastly, homologous cycles represent $k$-dimensional loops within the complex that can be continuously deformed into each other.

\begin{figure}[ht]
\centering
\includegraphics[width = 0.6\linewidth]{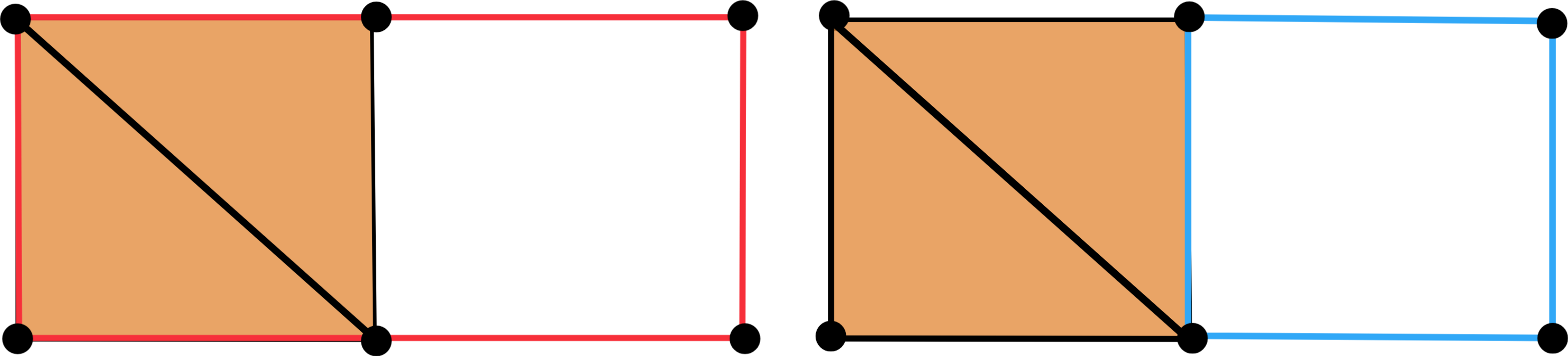}  
\caption{Two (red and blue) homologous 1-cycles}
\label{fig:homleft}
\end{figure}

The latter remark can be further visualized as follows. Consider two $1$-cycles colored red and blue in the simplicial complex shown in Figure \ref{fig:homleft}. These cycles may appear different, but if their difference can be expressed as the boundary of a $2$-simplex in the complex, they are considered homologous. This implies that these cycles can be continuously deformed, and any cycle may represent the homology class.

The generators of the homology group play a crucial role. They represent the distinct topological cycles embedded within our combinatorial model of the graph. Collectively, these cycles characterize the overall topology of the graph. Intriguingly, it is these generators and their corresponding homology classes that we leverage to define our centrality measures. However, a key point to remember is that equivalence within a homology class allows any generator to represent its class. For example, in Figure \ref{fig:homleft}, either the red or blue $1$-cycle can represent the same topological hole. This raises a natural question: How do we effectively choose representatives for these homology classes? This is where the concept of persistence comes into play, which we will explore in the next section.

\subsection{Filtrations and Persistent Homology}\label{subsec:pershom}

This section explores filtrations and persistent homology, powerful tools for analyzing topological features within simplicial complexes. These tools allow us to track how the underlying topological structure of a complex evolves as we progressively remove or modify its simplices based on some criteria.

We begin by introducing the concept of weights in simplicial complexes. A weight function $w:\mathscr{C} \to \mathbb{R}_{\geq 0}$ assigns weights (nonnegative real numbers) to each simplex $\sigma$ in the complex $\mathscr{C}$. This function allows us to differentiate between simplices and prioritize their removal during the filtration process. A common example of a weight function is the distance function, where the weight of a simplex might be defined as the maximum distance between any two points within that simplex. Higher weight values would then correspond to simplices with larger maximum distances, potentially representing less connected or more spread-out regions of the complex.

A filtration is a nested sequence of subcomplexes obtained by progressively removing simplices exceeding a certain weight threshold. Different weight order (e.g. $\geq$ order) leads to different filtration types. Intuitively, as the weight threshold increases, simplices with lower importance (based on the weight function) are removed. This reveals a coarser version of the original complex at each step, where only the most significant simplices remain.

This filtration process induces a sequence $\{\mathscr{C}_{w_i}\}$ of simplicial complexes. These complexes capture the connectivity structure of the original complex at each weight threshold. Persistent homology leverages this sequence to track the "persistence" of topological features (cycles) across different weight thresholds. In simpler terms, we are interested in how these features are born, die, or evolve as the filtration progresses.

\begin{definition}\label{def:pershom}
Let $\{\mathscr{C}_{w_i}\}$ be a filtered simplicial complex induced from a filtration of a simplicial complex $\mathscr{C}$. The \textit{$k$th $w_i-w_j$ persistent homology group} of $\mathscr{C}$, denoted by $H_k^{w_i, w_j}$, is the image of the induced homomorphism
$$
H_k(\mathscr{C}_{w_i}) \xrightarrow{\displaystyle{w_*}_k^{w_i, w_j}} H_k(\mathscr{C}_{w_j}).
$$
\end{definition}	

This group captures the homology classes "born" at weight $w_i$ and persisting (remaining in the homology group) until weight $w_j$. This group encodes how long homology classes survive under the filtration. For instance, a cycle that persists across a wide range of weight values suggests a more robust topological feature in the complex.

To facilitate future discussions, we introduce an abuse of notation where $\sigma$ can represent both a simplex and the corresponding cycle it defines. 

\begin{definition}\label{def:pers}
Suppose $\epsilon_0 < \epsilon < \epsilon^{\prime}$. We say that a homology class $[\sigma]$ is \textit{born} at $\mathscr{C}_{\epsilon}$ if $[\sigma] \in H_k^{\epsilon, \epsilon^{\prime}}$ but $[\sigma] \not\in H_k^{\epsilon_0, \epsilon^{\prime}}$. On the contrary, we say that $[\sigma]$ born at $\mathscr{C}_{\epsilon_0}$ \textit{dies} entering $\mathscr{C}_{\epsilon^{\prime}}$ if $w_k^{\epsilon_0, \epsilon^{\prime}}([\sigma]) = 0$ but $w_k^{\epsilon_0, \epsilon}([\sigma]) \neq 0$. Let $b(\sigma)$ and $d(\sigma)$, respectively, denote the birth and death of $[\sigma]$. We say that $\sigma_i$ \textit{gives birth to} the cycle $\sigma = \sum\sigma_i$ if $b(\sigma) = b(\sigma_i)$. 
\end{definition}

When two homology classes merge due to simplex removal, the elder rule comes into play. This rule selects the generator formed at the lower threshold as the natural representative of the merged class. A key observation is made: even when a different generator survives a merge, the persistence information can be "transferred" to the surviving one. This highlights the dynamics of merging classes and motivates further investigation into how these classes evolve under the filtration.

\begin{definition}\label{def:adjacent}
Let $\sigma$ and $\nu$ be $k$-cycles. We say that $\sigma$ and $\nu$ are \textit{$k$-near} if there exists a $k$-chain common to both $\sigma$ and $\nu$. Furthermore, the $k$th homology classes $[\sigma]$ and $[\nu]$ are \textit{$k$-near} if $\sigma$ and $\nu$ are $k$-near.
\end{definition}

The notion of $k$-nearness formalizes the idea that two cycles are close in the complex.

We also define the concept of merging governed by the elder rule. This definition clarifies the direction of merging, where a class with an earlier death threshold merges with the longer-lived class.

\begin{definition}\label{def:merge}
Let $[\sigma]$ and $[\nu]$ be distinct $k$th homology classes in a simplicial complex $\mathscr{C}_{\epsilon}$ where $d(\sigma) \neq d(\nu)$. We say that $[\sigma]$ and $[\nu]$ \textit{merge} at time $\epsilon^{\prime} = \min\left\{d(\sigma), d(\nu)\right\}$ if 
$$
\sigma + \nu = \rho
$$
for some $k$-boundary $\rho$.
\end{definition}

Within the framework of persistent homology, the merging of homology classes due to simplex removal is governed by the well-established \textit{elder rule}. This rule dictates that the generator formed at the lower filtration threshold becomes the designated representative for the newly formed merged class. Notably, the persistence information associated to the eliminated generator can still be "transferred" or "inherited" by the surviving one.

However, it is crucial to distinguish between the merging process and the act of gluing a boundary to a cycle. While merging signifies the disappearance of a topological feature upon simplex removal, gluing a boundary to a cycle essentially fills in the hole or closes the gap within the cycle. Although the resulting structure remains homologous, it falls outside the scope of current persistent homology methods. Addressing this distinction involves investigating the sequential gluing actions across the filtration. 

The dynamic nature of merging classes is underscored by the observed phenomenon of information transfer and the need to resolve boundary gluing. These observations prompt further inquiry into the evolution of these classes under the filtration process.

\subsection{Algorithm for Class Merges}

The computation of persistent homology often involves the algorithm introduced by Zomorodian \cite{phom}. This algorithm identifies the representatives (cycles) for each homology class. It mainly involves reducing the $k$th boundary matrix constructed by arranging the $(k - 1)$- and $k$-simplices in increasing births that comprise the rows and columns, and identifying the entries using the function
\[
\partial_k[i,j] = \begin{cases}1 & \text{ if } \sigma_i \text{ is a codimension-1 face of } \sigma_j \\ 0 & \text{ otherwise}\end{cases}.
\]  
The columns of the boundary matrix are ordered by birth.

A pseudocode of the algorithm \cite[p.~104]{ctop} is presented in Algorithm \ref{alg:phom}, where $l(i)$ denotes the position of the lowest $1$ in column $i$ of the boundary matrix.
\begin{algorithm}
\caption{Standard reduction algorithm for boundary matrices.}\label{alg:reduce}
\begin{algorithmic}[1]
\Require Boundary matrix of order $m$
\Ensure Reduced boundary matrix
\For{$j=2$ to $m$}
\While{there exists $i < j$ with $l(i) = l(j)$}
\State add column $i$ to column $j$
\EndWhile
\EndFor
\end{algorithmic}
\label{alg:phom}
\end{algorithm}

Let $c(\sigma)$ be the column associated with the simplex $\sigma$ in the boundary matrix. We introduce the ordering $\sigma \prec \delta$ to mean that either $b(\sigma) < b(\delta)$, or if $b(\sigma) = b(\delta)$ and $c(\sigma)$ appears first in the column of the boundary matrix. Additionally, let $r(\sigma)$ be the chain produced in $c(\sigma)$ by the reduction algorithm. 

\begin{lemma}\label{prop:unique}
Suppose that $[\sigma]$ is a $k$th homology class in a filtered simplicial complex $\mathscr{C}$. If $[\sigma]$ is not $k$-near to any homology class $[\nu]$, then the reduction algorithm produces a unique class representative for $[\sigma]$.
\end{lemma}

\begin{proof}
We assume that $\sigma = \sum_{i}\sigma_i$ and $\nu = \sum_{j}\nu_j$ are not $k$-near cycles where some $\sigma_i$ and $\nu_j$ respectively give birth to $\sigma$ and $\nu$. Furthermore, let $\nu_j \prec \sigma_i$. Note that the $k$th homology class representative depends on the associated $k$th boundary matrix, which is only concerned with simplices of dimensions $k - 1$ and $k$. Therefore, we can omit the case where $\sigma$ and $\nu$ has a common simplex of dimension $0$ to $k - 2$.

Suppose that $\sigma$ and $\nu$ have no shared $(k - 1)$-simplex. Then, for every simplex $\delta$, $l(\delta) \neq l(\sigma_i)$ for every integer $1 \leq i \leq n$. Therefore, the representative of $[\sigma]$ must be unique, since no simplex can be added with $\sigma_i$ to form another representative. 

Now, consider the case where $\sigma$ and $\nu$ intersect at some $(k - 1)$-faces. For the sake of contradiction, suppose that $\sigma$ has a representative containing some simplices from $\nu$. Since $\sigma$ and $\nu$ has no shared $k$-simplex, $r[\sigma]$ must contain all simplices of $\nu$. Otherwise, $r[\sigma]$ ceases to be a cycle if it does not contain all simplices of $\nu$. Note that, in the reduction algorithm, $c(\nu_j)$ reduces to a zero vector before reduction happens on $c(\sigma_i)$. Hence, $c(\sigma_i)$ cannot contain the simplex $\nu_j$ of $\nu$. Therefore, we arrive at a contradiction so $[\sigma]$ must have a unique class representative.
\end{proof}

Consider two $k$-near cycles $\sigma$ and $\nu$ that share a chain $\delta$. The shared simplices, along with simplices from $\sigma$ and $\nu$ that touch these shared simplices, can create situations where some simplices in $\sigma + \delta$ might have the same position of the lowest one as some simplices in $\nu + \delta$ and $\delta$. We call such connection points as \textit{junctures}.

\begin{lemma}\label{thm:reps}
Suppose that $[\sigma]$ and $[\nu]$ are distinct $k$th homology classes appearing in the filtration of a simplicial complex satisfying $\nu \prec \sigma$. Let $\nu = \sum_p\nu_p$ be a cycle representing $[\nu]$, where $\nu_p$ gives birth to $\nu$. Additionally, assume $\sigma$ and $\nu$ are $k$-near with common chain $\delta$. If $\nu_p \in \delta$ and $\nu$ does not contain simplices of $\sigma$ except $\delta$, then the representative of $[\sigma]$ produced by the reduction algorithm is given by $\sigma + \nu$. Otherwise, if $\nu_p \not\in \delta$ then the representative of $[\sigma]$ produced by the reduction algorithm is given by $\sigma$.
\end{lemma}

\begin{proof}
Suppose that $\sigma = \sum_{i}\sigma_i$ and $\nu$ are $k$-near with the intersection $\delta$, and $\nu_p \in \delta$. Let $\sigma_i$ be the simplex giving birth to $\sigma$. Note that $r(\sigma_i)$ only depends on the junctures and the algorithm proceeds by having $\sigma + \delta$ in $r(\sigma_i)$. If $\delta \in r(\sigma)$, then we arrive at a contradiction since $\nu_p \in \delta$ and $\nu_p$ gives birth to $\nu$ so $c(\nu_p)$ is a zero vector. 

If there exists no juncture where $\sigma_j \in \sigma + \delta$ and $\nu_q \in \nu$ such that $l(\sigma_j) = l(\nu_q)$, then $r(\sigma_i)$ ceases to be a cycle. Thus, there exists $\sigma_j \in \sigma + \delta$ and $\nu_q \in \nu$ such that $l(\sigma_j) = l(\nu_q)$. By the algorithm, $r(\nu_q) \in r(\sigma_i)$, and we have $\nu + \delta \in r(\sigma_i)$. Consequently, 
$$
r(\sigma_i) = (\sigma + \delta) + (\nu + \delta) = \sigma + \nu.
$$  
Therefore, we have proved the lemma.
\end{proof}

We are now able to provide a mechanism to track the merge information between homology classes and a criterion for when such happens. 

\begin{theorem}
\label{thm:merge}
Let $[\sigma]$ and $[\nu]$ be distinct persistent homology classes where $\sigma \prec \nu$. Then $[\sigma]$ merges with $[\nu]$ at $\epsilon = d(\sigma)$ if and only if $[\sigma]$ and $[\nu]$ are near.
\end{theorem}

\begin{proof}
Suppose $[\sigma]$ merges with $[\nu]$. If $[\sigma]$ and $[\nu]$ are not near, then both classes have unique cycle representatives that are not near. Thus, if $\sigma + \nu$ is a boundary at $\epsilon$ then $d(\sigma) = d(\nu)$. This equation is a contradiction to the assumption that the classes merge. Therefore, $[\nu]$ and $[\sigma]$ must be near.

Conversely, suppose that $[\sigma]$ and $[\nu]$ are near. By Lemma \ref{thm:reps}, we can choose $\sigma + \nu$ as representative of $[\sigma]$. Observe that 
$$
(\sigma + \nu) + \nu = \sigma.
$$ 
Since $\sigma \prec \nu$, $\sigma$ must be a boundary at $\epsilon$. Therefore, $[\sigma]$ merges with $[\nu]$. 
\end{proof}

In addition to keeping track of the cycle representatives yielded by the reduction algorithm, we are also interested in integrating the \emph{snowball effect} of successive merging classes in persistent homology. This necessitates the following definition.

\begin{definition}
Consider the induced homomorphism
$$
H_k(\mathscr{C}_{w_i}) \xrightarrow{\displaystyle{w_*}_k^{w_i, w_j}} H_k(\mathscr{C}_{w_j})
$$
and let $[\sigma]\in H_k(\mathscr{C}_{w_j})$. We define the \textit{first order merge cluster} of $[\sigma]$ at $w_j$ as
$$
M_1[\sigma, w_j] := \displaystyle({w_*}_k^{w_i, w_j})^{-1} ([\sigma]),
$$
that is, the set of homology classes merging with $[\sigma]$ at threshold $w_j$.
Inductively, for every integer $n \geq 2$, the \textit{$n$th order merge cluster} of $[\sigma]$ at $w_j$ is defined as 
$$
M_n[\sigma,w_j] = \displaystyle\bigcup_{\tau \in M_{n-1}[\sigma,w_j]} M_{1}[\tau,w_j].
$$
\end{definition}

The first-order merge clusters refer to all other classes that merges with a specific class $[\sigma]$ across the filtration until the threshold $w_j$. Higher-order merge clusters build on this idea. Imagine a class that merges with another class, which itself merges with a third class. All these classes are considered part of the same merge cluster because their merges are ultimately connected. The definition captures this cascading effect by recursively defining higher-order merge clusters as the union of all first-order merge clusters of classes within the previous-order merge cluster.

Theorem~\ref{thm:merge} offers a way to determine if two classes merge at a specific threshold. However, it might involve checking many pairs of cycles representing the classes, which can be computationally expensive. This is because any homologous cycle can represent a class, leading to a large number of potential pairings to explore. Corollary~\ref{cor:3.4} aims to reduce the number of pairings that need to be checked for merging. It leverages the concept of $k$-nearness between cycles representing classes.

\begin{corollary}
\label{cor:3.4}
Suppose that the homology class $[\sigma]$ merges with $[\nu]$. Let $[\delta]$ be a homology class such that $\sigma \prec \delta \prec \nu$. If $[\nu]$ and $[\delta]$ are near then either $[\delta]$ merges with $[\sigma]$, or there exists a homology class that merges with $[\nu]$ whose $n$th order merge cluster contains $[\delta]$ for some integer $n \geq 1$.
\end{corollary}

\begin{proof}
Suppose $[\nu]$ and $[\delta]$ are near. If $[\delta]$ merges with $[\nu]$, then we have proved the corollary. Suppose $[\delta]$ merges with a homology class $[\lambda]$ where $\delta \prec \lambda \prec \nu$. Note that a representative of $[\lambda]$ includes some simplices of $\delta$. Since $[\nu]$ and $[\delta]$ are near, $[\nu]$ and $[\lambda]$ must also be near. Therefore, $[\lambda]$ merges with $[\sigma]$.
\end{proof}

We can implement Theorem \ref{thm:merge} and Corollary \ref{cor:3.4} in an algorithm to obtain the first-order merge clusters of each homology class. Algorithm \ref{alg:merge} presents a pseudocode for finding the first order merge clusters.

\begin{algorithm}[ht]
\caption{Identifying First-Order Merge Clusters from Class Representatives.}
\label{alg:merge}
\begin{algorithmic}[1]
\Require Representatives $\left\{\sigma_i\right\}_{i = 1}^{l}$ ordered by ascending birth followed by death
\Ensure Collection containing first order merge clusters of the homology classes
\State Initialize an array $M$ of length $l$
\For{$i \in \{2,\dots,l\}$}
\For{$j \in \{1, \dots,i - 1\} \setminus \bigcup_{k = 1}^{i-1} M[k]$} 
\If{$[\sigma_i]$ is $k$-near to $[\sigma_j]$ and $d(\sigma_i) \geq d(\sigma_j) > b(\sigma_i)$} 
  \State insert $j$ to $M[i]$
\EndIf
\EndFor
\EndFor
\end{algorithmic}
\end{algorithm}

\section{Cycle Centrality}\label{sec:mainres}

This chapter explores the core of our study, defining centrality measures for cycles based on persistent homology. These measures go beyond a cycle's mere existence and capture its topological significance and influence within a network. We achieve this by considering two key aspects, namely persistence and merge dynamics. Our goal is to define centrality measures that are monotonic functions. This means that the importance of a cycle, as measured by our centrality, never decreases as the filtration progresses.

\subsection{Centrality Measures}

We propose three centrality functions, denoted by $J_1$, $J_2$, and $J_3$, to capture the evolving importance of cycles. Each function is defined for a specific homology class representative $[\sigma]$ and a filtration threshold $\epsilon$. All the functions share a common structure. They possess a base value of 0 before the birth threshold of $[\sigma]$, and exhibit a piece-wise linear behavior after birth, reflecting changes in importance.

The first centrality measure $J_1$ aims to capture and quantify the total persistence accumulated by $[\sigma]$ and all classes that directly merge with it up until the threshold $\epsilon$. Hence, if $P_{\epsilon}(\sigma)$ is the persistence of $\sigma$ at $\epsilon$, the first centrality function has the form

$$
J_1(\sigma, \epsilon) = \begin{cases} 0 & \text{for } \epsilon \leq b(\sigma) \\ 
P_{\epsilon}(\sigma) +\displaystyle\sum_{[\varsigma] \in M_1[\sigma, \epsilon]} P_{\epsilon}(\varsigma) & \text{for } \epsilon > b(\sigma) \\
\end{cases}.
$$

As each $P_{\epsilon}(\sigma)$ is monotonic and stabilizes when $\epsilon > d(\sigma)$, this function is piece-wise linear and monotonic. It captures the cascading simple aggregate of persistence pooled from cycle representatives that altogether merge directly to an older cycle. We view this as the homological importance of cycles --- if many cycles merge to an old cycle, then its homological significance is proportionally increased. 

One caveat of the function above is that it treats all merging instances similarly regardless of when the merge happens along the filtration. Hence, we also consider a second centrality function  $J_2$ by refining $J_1$ to account for the time when instances of merging happen. It introduces a scaling function that assigns a weight to the persistence of each merging class. This allows us to prioritize either early or late merges in the centrality calculation. We can write this function as 
$$
J_2(\sigma, \epsilon) = \begin{cases} 0 & \text{for } \epsilon \leq b(\sigma) \\ 
P_{\epsilon}(\sigma) + \displaystyle\sum_{[\varsigma] \in M_1[\sigma, \epsilon]} f_{\sigma}(\varsigma)P_{\epsilon}(\varsigma) & \text{for } \epsilon > b(\sigma) \\
\end{cases}.
$$
When early merges are considered more influential, we can define $f_{\sigma}(\varsigma)$ as $d(\varsigma)/d(\sigma)$. Conversely, when late merges are considered more important, we can define $f_{\sigma}(\varsigma)$ as $1-d(\varsigma)/d(\sigma)$.

We can also generalize the cascading effect of merging by considering the indirect transfer of persistence from merging instances prior to a given merge time. This is equivalent to modifying the centrality function to account for higher-order merging clusters, and yields a third centrality function $J_3$ given by
$$
J_3(\sigma, \epsilon) = \begin{cases} 0 & \text{for } \epsilon \leq b(\sigma) \\ 
P_{\epsilon}(\sigma) +\displaystyle\sum_r\displaystyle\sum_{[\varsigma] \in M_r[\sigma, \epsilon]} f_{\sigma}(\varsigma)P_{\epsilon}(\varsigma) & \text{for } \epsilon > b(\sigma) \\
\end{cases}.
$$

Here, we allow the definition $f_{\sigma}(\varsigma) = 1$ to generalize the function $J_1$.

In general, our centrality measures are of the form $J_n: \Lambda\times W\to \mathbb{R}_+$ where $\Lambda$ is the collection of all non-trivial persistent homology classes from a filtration. Hence, each persistence diagram produces a family of centrality functions. We can visualize the time-evolving centrality function of each homology class by constructing a piecewise-linear plot where sudden increases represent the effect of merging. We call these plots the $J_n$ \textit{centrality plots} of dimension $k$, where $k$ refers to the dimension of the homology classes we are concerned with. We can also create heat maps from contiguous heat bars each capturing the monotonic growth of centrality as captured by our measures. We illustrate a centrality plot in Figure \ref{fig:plotcircle} and omit the heat map representation. 

\begin{figure}[ht]
\centering
\begin{subfigure}{0.47\textwidth}
    \includegraphics[width=\linewidth]{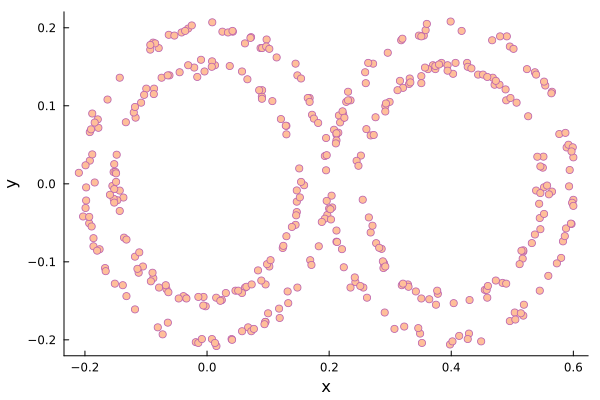}
    \caption{Point cloud}
    \label{fig:graph}
\end{subfigure}
\hfill
\begin{subfigure}{0.47\textwidth}
    \includegraphics[width=\linewidth]{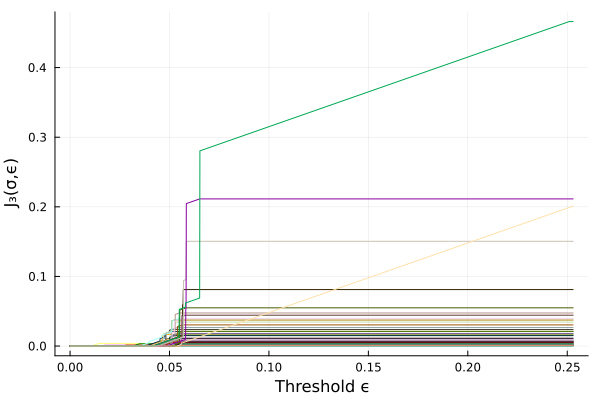}
    \caption{$J_3$ centrality plot}
    \label{fig:heat}
\end{subfigure}

\caption{The $J_3$ centrality plot, with $f_{\sigma} = 1$, of dimension $1$ produced by the Rips filtration of the point cloud sampled around a wedge sum of two annuli.}
\label{fig:plotcircle}
\end{figure}

\subsection{Stability of Centrality Measures}\label{subsec:stab}

This subsection investigates the stability of the centrality measures we defined earlier. Stability ensures that small changes in the network data, such as slight adjustments to edge weights, will not lead to drastic changes in the calculated centrality of cycles. 

First, we establish that all three centrality measures are monotonic. Intuitively, as we explore the network at a coarser scale, a hole might encounter more merging partners, potentially accumulating more persistence and thus increasing its centrality. In the following lemma, we denote $J_n(\sigma, \epsilon)$ by $J_{n,\sigma}(\epsilon)$.

\begin{lemma}[Monotonicity]\label{prop:mono}
If $\epsilon < \epsilon^{\prime}$, then $J_{n,\sigma}(\epsilon) \leq J_{n,\sigma}\left(\epsilon^{\prime}\right)$.
\end{lemma}

\begin{proof}
It follows from definition that $M_k[\sigma, \epsilon] \subseteq M_k[\sigma, \epsilon^{\prime}]$ for $\epsilon < \epsilon^{\prime}$ and $k \geq 1$. Since $P_{\epsilon}$ is monotonic and $f$ is constant with respect to $\epsilon$, the conclusion follows.
\end{proof}

To analyze stability, we introduce the concept of a $p$-centrality norm. This norm assigns a numerical value to each hole's centrality, capturing the overall importance of that hole. The specific value of $p$ influences how the norm prioritizes different aspects of the centrality function. For example, $p=1$ emphasizes the average centrality, while $p= \infty$ focuses on the maximum centrality value achieved by the hole. 

\begin{definition}
    Let $\mathcal{J}_{n} = \{J_{n,\sigma}|[\sigma]\in \Lambda\}$ denote the collection of centrality functions generated by the set of persistent homology classes $\Lambda$. The \emph{p-centrality norm} given by  
$$
\left\|J_{n,\sigma}(d^*)\right\|_p = \begin{cases} \left(\int_0^{d^{*}}(J_{n,\sigma}(x))^p \ dx\right)^{1/p} & \mbox{ if } 1 \leq p < \infty\\
J_{n,\sigma}(d^* & \mbox{ if } p = \infty
\end{cases}
$$
where $d^*$ is the minimum between $d(\sigma)$, and the largest geodesic distance between any two vertices in the largest cycle in $[\sigma]$.
\end{definition}

Next, we define a metric called the $p$-centrality distance to compare collections of centrality norms. Since a natural order ceases to exist for the centrality measures, we appeal to a bottleneck-like distance approach. The computation considers optimal pairings between centrality norms in the two collections and focuses on the maximum difference between any two paired values. For computational efficiency, we match the condition in the bottleneck distance implementation in \cite{lumawig}. In this case, let $\Omega = 0\times \{\|J_{n,\sigma}(d^*)\|_p^p : J_{n,\sigma} \in \mathcal{J}_{n}\}$ and $\Omega^{\prime} = 0\times \{\|J_{n,\sigma^{\prime}}({d^{*}}^{\prime})\|_p^p : J_{n,\sigma^{\prime}} \in \mathcal{J}^{\prime}_{n}\}$. For $x_{\sigma}\in \Omega$, we define $\delta_{x_{\sigma}} = \left\|J_{n,\sigma}\right\|_p^p$ . For a bijection $\phi:\Omega\cup \Delta \to \Omega^{\prime}\cup \Delta$, define 
\begin{equation}\label{eqn:bottle}
\|x_{\sigma} - \phi(x_{\sigma})\|_{\infty} = \begin{cases}\frac{1}{2}\max\{\delta_{x_{\sigma}}, \delta_{\phi(x_{\sigma})}\} & \text{if } \phi(x_{\sigma}) \in \Delta \\ |\delta_x - \delta_{\phi(x)}| & \text{otherwise} \end{cases}.
\end{equation}

\begin{definition}
For $1 \leq p < \infty$, the \textit{$p$-centrality distance} is given by 
$$
C_p\left(\mathcal{J}_{n}, \mathcal{J}_{n}^{\prime}\right) = \inf_{\phi}\sup_{x \in X}\|x - \phi(x)\|_{\infty}
$$
where the infimum is taken over all bijections from $\Omega\cup \Delta$ to $\Omega^{\prime}\cup \Delta$.    
\end{definition}

For the case where $p = \infty$, we propose a distance akin to $p$-landscape distance \cite[p.~94]{landscape}. Note that we can order the centrality function $J_{n,\sigma}(d*)$ of each homology class $[\sigma]$ based on the maximum centrality values. Thus, we obtain a decreasing sequence $\{\|J_{n,m}\|\}_m$ where $m$ is a positive integer. The \textit{$p$-centrality distance} is then given by
$$
C_p\left(\mathcal{J}_{n, k}, \mathcal{J}_{n, k}^{\prime}\right) = \sum_m\left\|J_{n,m} - J_{n,m}^{\prime}\right\|_{p}.
$$

\begin{example}\label{example_stab}
We examine how our proposed centrality measures behave with respect to perturbations of the point cloud in Figure \ref{fig:graph} introduced by replacing each point $(x,y)$ with $(x + \kappa_1, y + \kappa_2)$ for some $\kappa_1, \kappa_2 \in [-\kappa, \kappa]$.
We then compute the $1$-centrality distance between the centrality measures of the original point cloud and its perturbation, and replicate thirty simulations of this process to generate a distribution of $1$-centrality distances represented by boxplots. In Figures \ref{fig:boxplot}, we observe how the distribution of $1$-centrality distances varies across increasing levels of perturbations.  

\begin{figure}[ht]
\centering
\includegraphics[width = 0.98\textwidth]{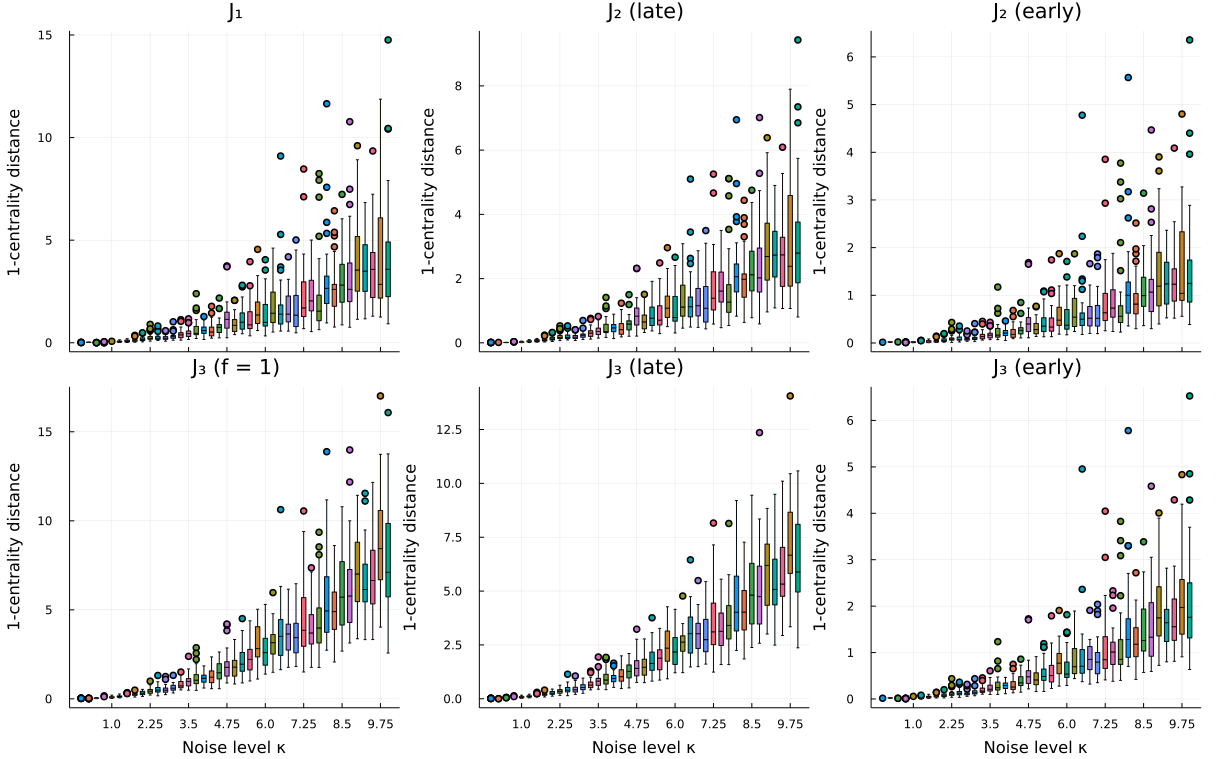}
\caption{Boxplots for the $1$-centrality distance between the centrality functions of the point cloud in Figure \ref{fig:graph} and its perturbations for all noise levels.}
\label{fig:boxplot}
\end{figure}
\end{example}

We now establish bounds for the $p$-centrality distance to quantify the stability of the centrality measures. We simplify notations by dropping $n$ and write $J_{\sigma}$ when we consider any of the centrality functions $J_{1,\sigma}$, $J_{2,\sigma}$, and $J_{3,\sigma}$. 

The following results focus on collections $\Lambda$ and $\Lambda^{\prime}$ of persistent homology class representatives obtained through a filtration of a simplicial complex. We define $K$ as the maximum persistence observed across all cycles in both $\Lambda$ and $\Lambda^{\prime}$. Additionally, we define $q$ as the maximum number of cycles between $\Lambda$ and $\Lambda^{\prime}$. This expression represents the larger collection size for comparison purposes. Lastly, we assume that there exists a homology class with a non-empty first-order merge cluster.

\begin{lemma}\label{prop:finite}
Let $\mathcal{J}(\Lambda)$ be the collection of centrality measures generated from $\Lambda$. Then
$$
C_p\left(\mathcal{J}(\Lambda), \mathcal{J}\left(\Lambda^{\prime}\right)\right) \leq
\begin{cases}
K^{1 + \frac{1}{p}}\left(1 + q\right) &\mbox{ if } 1 \leq p < \infty\\
Kq(1 + q)&\mbox{ if }p = \infty
\end{cases}.
$$
\end{lemma}

\begin{proof}
Monotonicity of the centrality measures (Proposition \ref{prop:mono}) yields
$$
J_{\sigma}(\epsilon) \leq P_{\epsilon}(\sigma) + \sum_r\sum_{[\nu] \in M_r[\sigma, d(\sigma)]} f_{\sigma}(\varsigma)P_{\epsilon}(\nu).
$$
By the definitions of $K$ and $q$, we have $J_{\sigma}(\epsilon) \leq K + f_{\sigma}(\varsigma)qK$. Thus, by the monotonicity of Lebesgue integrals, we obtain $\|J_{\sigma}(\epsilon)\|_p \leq \sqrt[p]{d^*}\left[K + f_{\sigma}(\varsigma)qK\right]$ where the integral is taken from $[0, d^*]$. Recall that $f_{\sigma}(\varsigma) \in (0,1]$. Thus,
\begin{equation}\label{eqn:1}
    \|J_{\sigma}(\epsilon)\|_p \leq K^{1 + \frac{1}{p}}(1+q).
\end{equation}

Now, we define $X = \{\|J_{\sigma}\|: J_{\sigma} \in \mathcal{J}(\Lambda)\}$ and $Y = \{\|J_{\delta}^{'}\| : J_{\delta}^{'} \in \mathcal{J}(\Lambda^{\prime})\}$. Assuming $|\Lambda| \leq |\Lambda^{\prime}|$, we consider a bijection $\phi: X \cup \Delta \rightarrow Y$. Note that $\left|\|J_{\sigma}\|_p-\|J_{\delta}\|_p\right| \leq \max\{\|J_{\sigma}\|_p, \|J_{\delta}\|_p\}$. Using Equation \ref{eqn:1}, for any $x \in X$ and $y \in Y$,
$$
\left|x - y\right| \leq K^{1 + \frac{1}{p}}\left(1 + q\right).
$$
Now, suppose $x \in \Delta$. By Equation \ref{eqn:bottle}, we obtain
$$
\dfrac{1}{2}\max\{x, \phi(x)\} \leq \dfrac{1}{2}K^{1 + \frac{1}{p}}\left(1 + q\right) < K^{1 + \frac{1}{p}}\left(1 + q\right).
$$
Therefore, $C_p\left(\mathcal{J}(\Lambda), \mathcal{J}\left(\Lambda^{\prime}\right)\right) \leq K^{1 + \frac{1}{p}}\left(1 + q\right)$ for $1 \leq p < \infty$. 

We consider the case where $p = \infty$. Note that, for any $\epsilon \geq 0$, $J_{\sigma}(\epsilon) \geq 0$. Hence,
$$
\sum_m\left\|J_m - J_m^{\prime}\right\|_{\infty} \leq \sum_m\max\{J_m(d*), J_m(d*)\}.
$$
Since $J_{\sigma}(\epsilon) \leq K + f_{\sigma}(\varsigma)qK$ and $f_{\sigma}(\varsigma) \in (0,1]$, we have 
$$
\sum_m\left\|J_m - J_m^{\prime}\right\|_{\infty} \leq \sum_mK(1 + q).
$$ 
By applying the definition of $q$, we have proven the lemma.
\end{proof}

The value of $K$ suggests a potential for variations in the lifetime of topological features. In the previous lemma, a smaller value of $K$ contributes to tighter bounds.

We now introduce a constant that will be instrumental in establishing an upper bound for the centrality distance. This constant is given by
$$
R(p) := 
\begin{cases}
\sqrt[p]{2}K(1 + q) & \text{if } 1 \leq p < \infty \\
2q(1 + q) & \text{if } p =\infty
\end{cases}. 
$$
The next step is to demonstrate that the $p$-centrality distance between any two collections of centrality measures is always upper bounded by this constant and the bottleneck distance between the corresponding collections.

\begin{theorem}\label{thm:bot}
Let $D$ and $D^{\prime}$ represent the persistence diagrams corresponding to the collections $\Lambda$ and $\Lambda^{\prime}$. Then
$$
C_p\left(\mathcal{J}(\Lambda), \mathcal{J}\left(\Lambda^{\prime}\right)\right) \leq
\begin{cases}
R(p)\sqrt[p]{d_B(D, D^{\prime})} &\mbox{ if } 1 \leq p < \infty\\
R(p)d_B(D, D^{\prime})&\mbox{ if }p = \infty
\end{cases}.
$$
\end{theorem}

\begin{proof}
For any pair of homology classes $[\sigma]$ and $[\delta]$, application of the triangle inequality yields
$$
|P_{\epsilon}(\sigma) - P_{\epsilon}(\delta)| \leq |d(\sigma) - d(\delta)| + |b(\sigma) - b(\delta)|.
$$
By the definition of the bottleneck distance, we have $|P_{\epsilon}(\sigma) - P_{\epsilon}(\delta)| \leq 2d_B(D, D^{\prime})$. For the left-hand expression, the maximum taken over all homology classes $[\sigma]$ and $[\delta]$ is $P_{\epsilon}(\sigma)$ or $P_{\epsilon}(\delta)$. Consequently, $K \leq 2d_B(D, D^{\prime})$. The conclusion follows from Lemma \ref{prop:finite}.
\end{proof}

The following corollaries leverage the combinatorial stability theorem \cite[p.~123]{stability} to reformulate the previously established bounds in terms of the constant $R(p)$ and properties of the filtration functions. We omit the proofs for these corollaries as they directly apply the referenced theorem.

\begin{corollary}\label{cor:stab}
Let $w, w^{\prime} :\mathscr{C} \rightarrow \mathbb{R}$ be monotone real-valued functions that filter the simplicial complex $\mathscr{C}$. Then
$$
C_p\left(\mathcal{J}(\Lambda), \mathcal{J}\left(\Lambda^{\prime}\right)\right) \leq
\begin{cases}
R(p)\sqrt[p]{\|w - w^{\prime}\|_{\infty}} &\mbox{ if } 1 \leq p < \infty\\
R(p)\|w - w^{\prime}\|_{\infty}&\mbox{ if }p = \infty
\end{cases}.
$$
\end{corollary}

In the next corollary, we introduce another constant given by
$$
R^{\prime}(p) := 
\begin{cases}
\sqrt[p]{2}K(1 + q^{\prime}) & \text{if } 1 \leq p < \infty \\ 
2q^{\prime}\left(1+q^{\prime}\right) & \text{if } p = \infty
\end{cases}
$$
where $q^{\prime} := \max\left\{\sum_r|M_r[\sigma, d(\sigma)]|: \sigma \in \Lambda \cup \Lambda^{\prime} \text{ and } P_{\epsilon}(\sigma) \neq 0\right\}$.

The term $q'$ represents the maximum number of successive merging of homology classes across $\Lambda$ and $\Lambda'$.

\begin{corollary}\label{thm:sharp}
Let $D$ and $D^{\prime}$ be persistence diagrams corresponding to the collections $\Lambda$ and $\Lambda^{\prime}$ and let 
 Then 
$$
C_p\left(\mathcal{J}_n(\Lambda), \mathcal{J}_n\left(\Lambda^{\prime}\right)\right) \leq \begin{cases}R^{\prime}(p)\sqrt[p]{d_B(D, D^{\prime})} & \text{if } 1 \leq p < \infty \\ R^{\prime}(p)d_B(D, D^{\prime}) & \text{if } p = 
\infty\end{cases}.
$$ 
\end{corollary}

The previous corollary refines the bounds further by incorporating information about the merge clusters of homology classes. A smaller $q'$ indicates less variation in how homology classes merge between $\Lambda$ and $\Lambda'$, leading to potentially tighter bounds.

\begin{theorem}[Stability]\label{prop:stab1}
Let $w, w^{\prime} :\mathscr{C} \rightarrow \mathbb{R}$ be monotone real-valued functions that filter the simplicial complex $\mathscr{C}$. The inequality
$$
C_p\left(\mathcal{J}_n(\Lambda), \mathcal{J}_n\left(\Lambda^{\prime}\right)\right) \leq R^{\prime}(p)\|w - w^{\prime}\|_{\infty}.
$$ 
holds when either $p = \infty$, or $1 \leq p < \infty$ and $q^{\prime} > \left(1/\sqrt[p]{2}K\right) - 1$. 
\end{theorem}
\begin{proof}
Suppose that $1 \leq p < \infty$ and $q^{\prime} > \left(1/\sqrt[p]{2}K\right) - 1$. It follows from the combinatorial stability theorem and Corollary \ref{thm:sharp} that
$$
\dfrac{1}{R^{\prime}(p)}C_p\left(\mathcal{J}_n(\Lambda), \mathcal{J}_n\left(\Lambda^{\prime}\right)\right) \leq \|w - w^{\prime}\|_{\infty}^{\frac{1}{p}}.
$$
The conclusion follows because $x^{\frac{1}{p}} \leq x$ for any $x \geq 0$. 

Suppose $p = \infty$. The inequality $2q^{\prime}(1+q^{\prime}) \geq 1$ holds since $q^{\prime} \geq 1$. Applying the combinatorial stability theorem, we obtain
$$
\dfrac{1}{R^{\prime}(p)}C_p\left(\mathcal{J}_n(\Lambda), \mathcal{J}_n\left(\Lambda^{\prime}\right)\right) \leq \|w - w^{\prime}\|_{\infty}.
$$
\end{proof}

To evaluate the effectiveness of the bounds established in Corollary~\ref{thm:sharp}, we perform the following analysis. We consider the 1-centrality distances (refer to Example~\ref{example_stab}) calculated for various perturbation levels. From the bounds provided by Corollary~\ref{thm:sharp}, we subtract the actual 1-centrality distances. These differences are then plotted as boxplots across increasing levels of perturbations, similar to the approach used in the previous figures. This visualization allows us to assess how closely the theoretical bounds align with the empirical observations of the 1-centrality distance under increasing network perturbations.

\begin{figure}[ht]
\centering
\includegraphics[width = 1.00\textwidth]{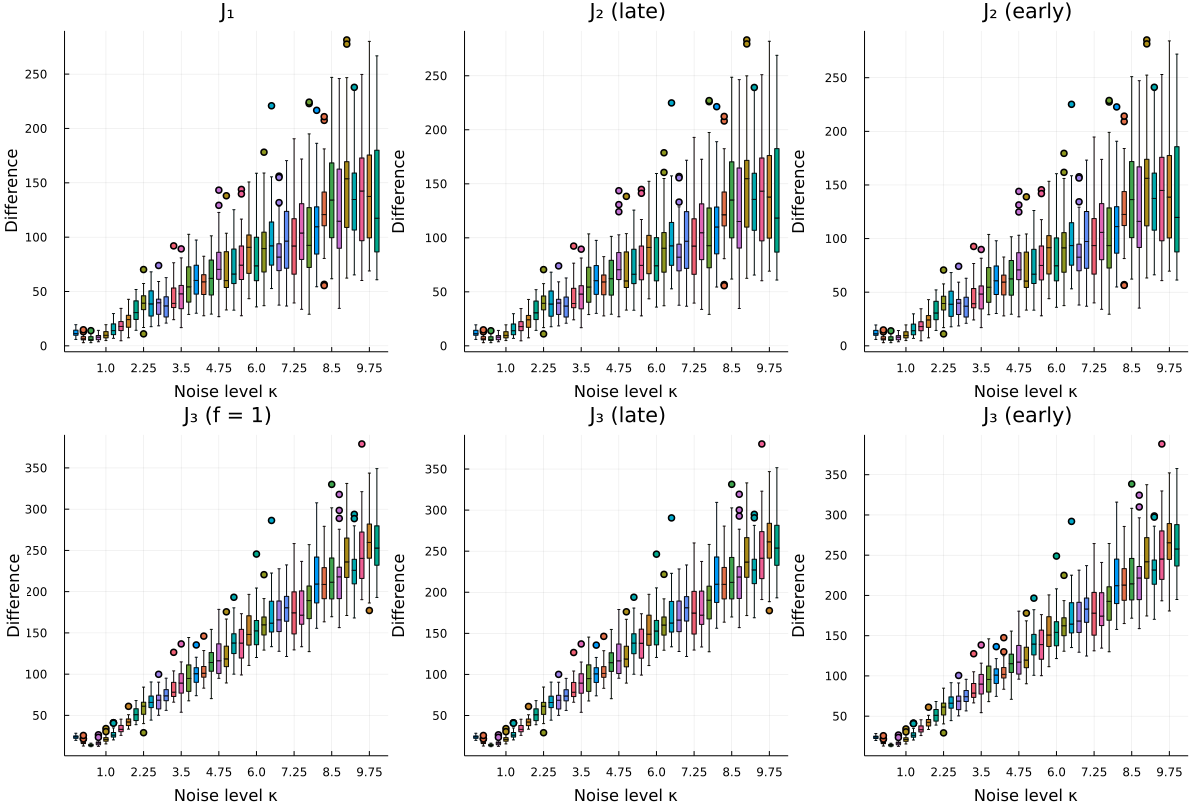}
\caption{Boxplots of the difference between the $1$-centrality distance in Figure \ref{fig:boxplot} and corresponding bounds given by Proposition \ref{thm:sharp}.}
\label{fig:stabbp}
\end{figure}

\section{Application to Fractal-Like Point Clouds}

This chapter explores the application of centrality measures to self-similar point clouds. While toy datasets offer valuable starting points, their applicability to real-world scenarios may be limited. In contrast, fractal-like point clouds, with their inherent complexity and potential for higher dimensionality, have been vastly documented to appear in nature \cite{fractal}. Persistent homology has been studied on fractals such as computing the affine fractals from landscapes \cite{fractalpers} and estimating the fractal dimension \cite{fractaldim}. Analyzing these structures allows us to explore the generalizability and effectiveness of centrality measures in deciphering intricate network-like relationships within spatial data.

For this application, we employ a method for separating signal from noise in persistence diagrams using a hypothesis testing approach developed by Bobrowski and Skraba \cite{universal}. This method relies on the concept of \textit{multiplicative persistence} $(\pi(p) = \frac{d}{b})$ for any birth-death pair $(b,d)$ in the $k$th persistence diagram $\text{dgm}_k$. Throughout this chapter, we operate under the assumption of the validity of the strong universality conjecture. Within this theoretical framework, a left-skewed Gumbel distribution (LGumbel) plays a pivotal role in the hypothesis testing process. In the conjecture below, the constant $\lambda$ is the Euler-Mascheroni constant. Moreover, $A(\mathcal{T})$ is $1$ if $\mathcal{T}$ is the Rips filtration, and $0.5$ if $\mathcal{T}$ is the \v{C}ech filtration.

\begin{conjecture}
Consider $d$-dimensional metric measure space $\mathcal{S}$ and a sequence of random variables $\mathbf{X}_n = (X_1,\dots,X_n) \in S^n$ with joint probability law $\mathbb{P}_n$. Let $\mathbb{S} = (\mathcal{S},\mathbb{P})$ be a sampling model under a filtration type $\mathcal{T}$. For any $\mathbb{S} \in \mathcal{U}, \mathcal{T}$, and $k \geq 1$, the limit of $\mathcal{L}_n$ as $n$ approaches infinity  equals the left-skewed Gumbel distribution, where
\[
\mathcal{L}_n(\mathbb{S}, \mathcal{T}, k) = \frac{1}{|\text{dgm}_k|}\sum_{p\in\text{dgm}_k}\delta_{l(p)}
\]
and $l(p) = A(\mathcal{T})\log\log(\pi(p)) - \lambda - \bar{L}$. The expression $\bar{L}$ is given by
\[
\frac{1}{\text{dgm}_k}\sum_{p \in \text{dgm}_k}\log\log(\pi(p)).
\]
\end{conjecture}

This conjecture states that, as the sample size approaches infinity, the distribution of certain features derived from persistence diagrams converges to the left skewed Gumbel (LGumbel) distribution.

From the hypothesis testing framework, we recover the signal part $\text{dgm}_k^S$ of the diagram by considering only those $p \in \text{dgm}_k$ where $$e^{-e^{l(p)}} < \frac{\alpha}{|\text{dgm}_k|}$$ for a significance level $\alpha$. Until the end of this chapter, we consider the significance level $\alpha = 0.05$.

We now explore the topological structure of a point cloud containing 1,000 points sampled around the well-known Sierpinski Triangle (Figure~\ref{fig:sierp}). Given the inherent variability in point cloud data and the complex nature of TDA, we employ a bootstrapping approach to assess the robustness of our findings regarding the number of holes (signals) identified. 

\begin{figure}[ht]
\centering
\includegraphics[width=0.45\linewidth]{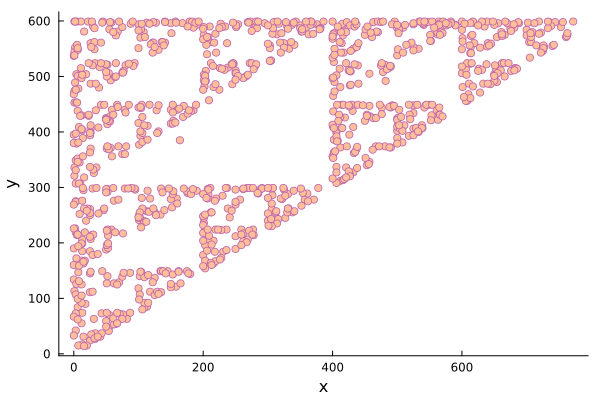}
\caption{A point cloud sampled around a Sierpinski triangle.}
\label{fig:sierp}
\end{figure}

Bootstrapping \cite{bootstrap} involves creating multiple random samples (with replacement) from the original point cloud. In this case, we generated 1,000 samples, each containing 800 points (80\% of the original data). This sampling percentage is chosen following the rule of thumb that using at least 50\% of the data is beneficial for obtaining statistically significant results in bootstrapping.

We then use the hypothesis testing framework for each bootstrapped sample to identify the number of holes. The mean number of holes across all bootstrapped samples is $0.68$, with an approximate standard error of $0.018$. The 95\% confidence interval for the number of holes is $(0.643, 0.717)$. The confidence interval suggests that the true number of holes in the original point cloud likely falls between 0 and 1. However, the interval is skewed slightly closer to 1, indicating a higher probability of there being a single hole present in the data.

We examine one bootstrapped sample whose plot is shown in Figure~\ref{fig:sierp_boot}. Applying the hypothesis testing framework, we identify only a single hole as statistically significant (signal). The corresponding signal is visualized in Figure~\ref{fig:sierp_sig}. The signal contains points from the largest triangle in the fractal. 

\begin{figure}[ht]
\centering
\begin{subfigure}{0.44\textwidth}
    \includegraphics[width=\linewidth]{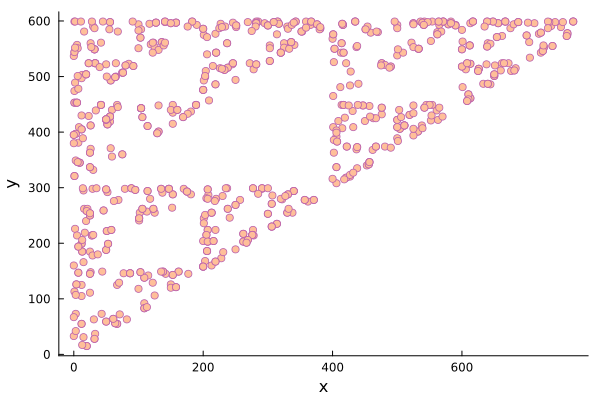}
    \caption{Point cloud}
    \label{fig:sierp_boot}
\end{subfigure}
\hfill
\begin{subfigure}{0.44\textwidth}
    \includegraphics[width=\linewidth]{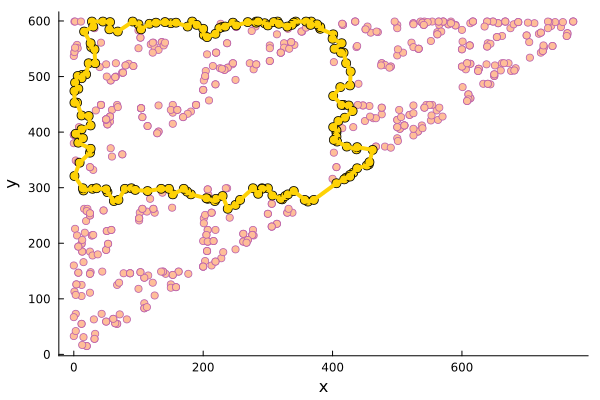}
    \caption{Signal (yellow)}
    \label{fig:sierp_sig}
\end{subfigure}
\caption{Signal from a bootstrap sample of the point cloud in Figure \ref{fig:sierp}.}
\end{figure}

In the context of centrality, we consider two primary ideas for analysis. Direct merging events between holes offer insights into the localized dynamics around the merged region. However, a more comprehensive perspective is gained by considering the cumulative effect of merging across the entire filtration process. This broader analysis, encompassing all merging events, can reveal structures with global importance in shaping the overall landscape of significant holes. Next, while centrality measures rank the holes, a quantitative method to pinpoint the most important ones might not be readily available. Here, we address this challenge by examining the relative difference in centrality values between the top-ranked holes and the majority of others. Significant differences in centrality would suggest a higher likelihood of those top-ranked holes being truly central features within the data.

Figure~\ref{fig:sierp_cent} depicts the centrality plots associated with the point cloud. To distinguish between our different centrality measures, we denote centrality functions $J_3$ with various scaling factors as $J_4$, $J_5$, and $J_6$. The function $J_4$ has factors equal to one, while the functions $J_4$ and $J_5$ prioritize late merging and early merging, respectively. 

\begin{figure}[ht]
\centering
\includegraphics[width=0.5\linewidth]{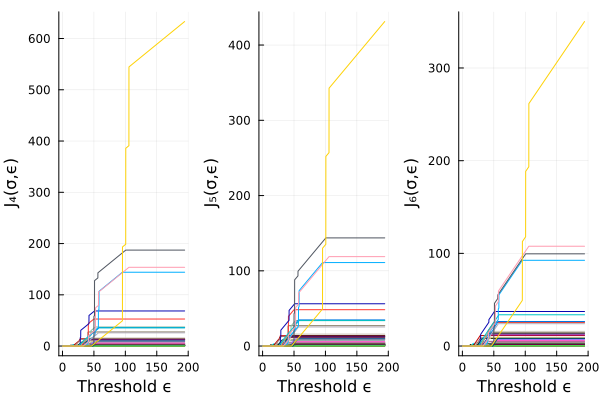}
\caption{Centrality plots of dimension $1$ associated to the Rips filtration of the point cloud in Figure \ref{fig:sierp}.}
\label{fig:sierp_cent}
\end{figure}

The centrality plots suggest the presence of a relatively important signal within the point cloud data. This is evidenced by the large difference in the maximum centrality value for the highest ranked hole compared to the others. Notably, this hole coincides with the one previously identified using the hypothesis testing framework. This strong agreement between centrality measures and the well-established topological tools found in persistence diagrams reinforces the significance of this particular hole as a key topological feature within the data. 

Further bolstering this observation, we can consider persistence values as a ranking system for holes. Computing the Spearman rank correlation coefficient \cite{spearman} between the maximum centrality values of $J_5$ and the persistence values yields a value of 0.997. This value indicates a near-perfect, monotonically increasing relationship between the two rankings. In simpler terms, holes ranked highly by centrality also tend to have high persistence values. By aligning with established methods like persistence diagrams, these findings suggest that centrality functions effectively capture features similar to those identified by common TDA summaries.

We examine other bootstrap sample (Figure~\ref{fig:sierp_boot1}) where the hypothesis testing framework fails to identify a signal. Interestingly, the persistence diagram (Figure~\ref{fig:sierp_diag}) exhibits a point far from the diagonal, potentially indicating a feature. However, this point is not classified as a signal by the hypothesis testing framework.

Despite the lack of a signal using traditional methods, both the persistence and the centrality functions, viewed as rankings, reveal some patterns. Mirroring the previous bootstrap sample, the highest ranked hole by both methods corresponds to points from the largest triangle in the fractal. Similarly, the second highest ranked hole aligns with points from the next largest triangle.

\begin{figure}[ht]
\centering
\begin{subfigure}{0.48\textwidth}
    \includegraphics[width=\linewidth]{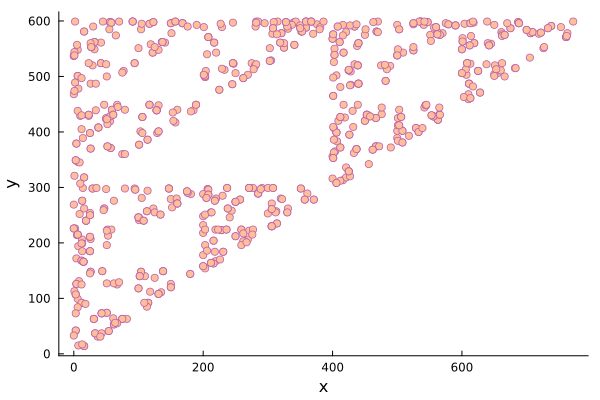}
    \caption{Point cloud}
    \label{fig:sierp_boot1}
\end{subfigure}
\hfill
\begin{subfigure}{0.48\textwidth}
    \includegraphics[width=\linewidth]{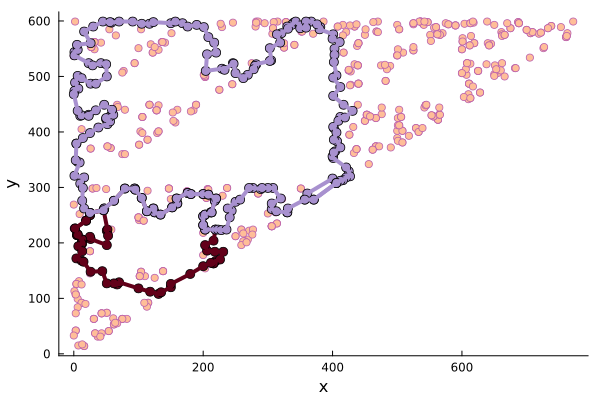}
    \caption{Two highest ranked holes}
    \label{fig:sierp_sig1}
\end{subfigure}
\hfill
\begin{subfigure}{0.48\textwidth}
    \centering
    \includegraphics[width=\linewidth]{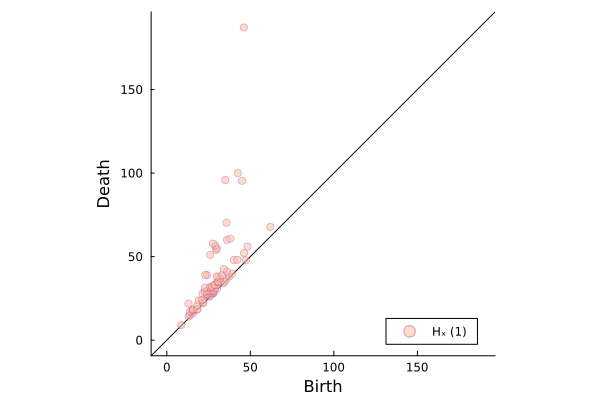}
    \caption{Persistence diagram of dimension $1$}
    \label{fig:sierp_diag}
\end{subfigure}
\hfill
\begin{subfigure}{0.48\textwidth}
    \centering
    \includegraphics[width=\linewidth]{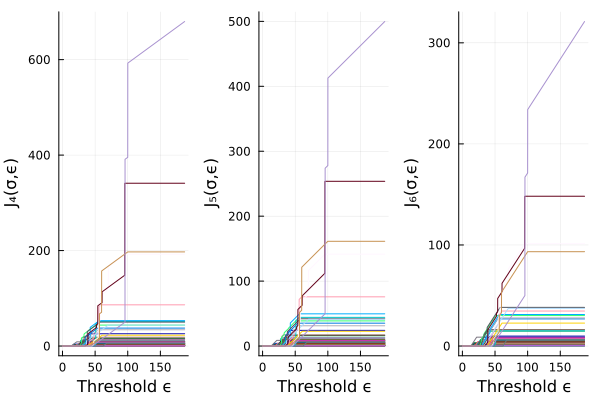}
    \caption{Centrality plots of dimension $1$}
    \label{fig:sierp_boot_cent}
\end{subfigure}
\caption{Two highest ranked holes according to the persistence diagram and the centrality plots from a bootstrap sample of the point cloud in Figure \ref{fig:sierp} with no identified signals.}
\end{figure}

To gain a deeper understanding, we delve into the values of both summaries and how they identify holes. We examine how many holes satisfy a threshold, defined as a percentage ($i$) of the maximum value from each summary. We will analyze a range of percentages for $i$ from $0.25$ to $1$. This analysis will be performed for both the persistence diagram and the centrality function $J_5$. Figure~\ref{fig:sierp_thresh} illustrates the graph of the number of holes identified by each method at different thresholds. The graph reveals interesting differences in how persistence and centrality identify holes. While $J_5$ identifies two holes meeting the threshold at $i = 0.5$, the persistence diagram identifies a single hole at a slightly lower threshold (around $0.4$). As we move to lower thresholds, the persistence diagram consistently identifies more holes than $J_5$ at each step. These observations suggest that centrality has the tendency to be selective in identifying features compared to the persistence diagram. These observations highlight the potential of centrality functions to capture features that might differ when using common topological summaries.

\begin{figure}[ht]
\centering
\begin{subfigure}{0.48\textwidth}
    \includegraphics[width=\linewidth]{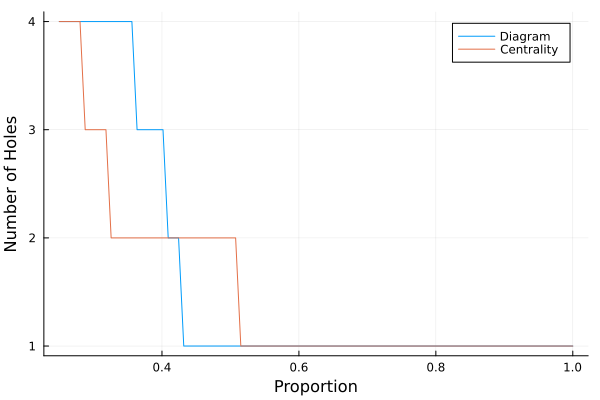}
    \caption{Number of holes satisfying the threshold}
    \label{fig:sierp_thresh}
\end{subfigure}
\hfill
\begin{subfigure}{0.48\textwidth}
    \includegraphics[width=\linewidth]{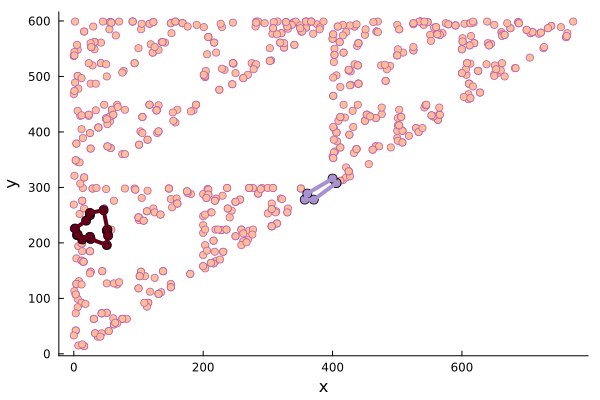}
    \caption{Earliest cycles across all orders of merge clusters}
    \label{fig:short}
\end{subfigure}
\caption{Additional features that the centrality functions capture from the point cloud in Figure \ref{fig:sierp_boot1}.}
\end{figure}

We can also identify and analyze the merging events across the entire filtration process. Figure~\ref{fig:short} shows the earliest homology class considering all possible orders of merging clusters in each of the two highest ranked holes. We note the small diameter of each hole. This observation suggests that some short-lived holes, present during the early stages of filtration, contribute to the overall centrality of the larger holes. Additionally, the centrality of these large holes appears to be influenced by merge instances originating from them much earlier in the filtration process.

\begin{figure}[ht]
\centering
\begin{subfigure}{0.49\textwidth}
    \centering
    \includegraphics[width=\linewidth]{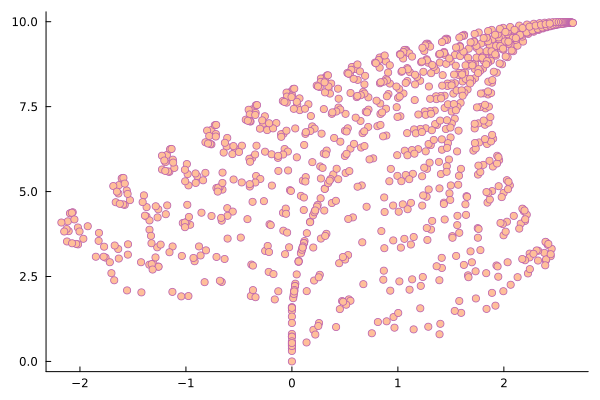}
    \caption{Point cloud}
    \label{fig:fern}
\end{subfigure}
\hfill
\begin{subfigure}{0.49\textwidth}
    \includegraphics[width=\linewidth]{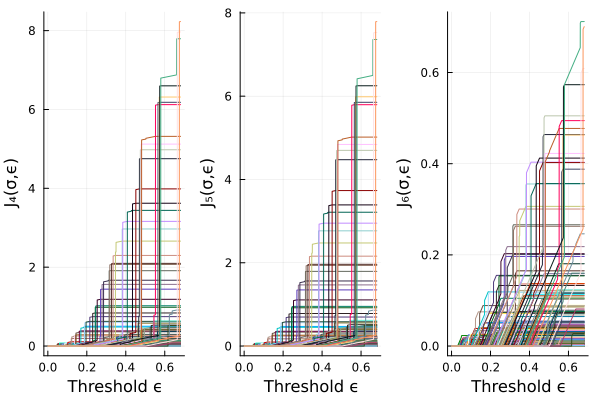}
    \caption{Centrality plots of dimension $1$}
    \label{fig:fern_diag}
\end{subfigure}
\caption{Centrality plots, of dimension $1$, produced by the Rips filtration of a point cloud sampled around the Barnsley Fern.}
\end{figure}

We conclude this section by examining a different point cloud, sampled around the Barnsley Fern fractal (Figure~\ref{fig:fern}). Here, we employ the same bootstrapping approach used previously. No signals were identified in any of the bootstrapped samples for the Barnsley Fern. The Spearman rank correlation between the maximum $J_5$ centrality values and the persistence values remains high (0.893), indicating a strong monotonically increasing relationship. Similar to the previous example, this suggests that centrality and persistence rankings tend to agree on the relative importance of holes despite the absence of identifiable signals. Furthermore, Figure~\ref{fig:map_fern} shows the three highest ranked holes identified by centrality. Extracting the earliest cycles across all possible orders of merge clusters for these three holes reveals a single cycle. In simpler terms, this implies that one hole merged with another, then the survivor merged with the third hole, forming a single connected component at some point during the filtration process. The representative of the single cycle further suggests that the corresponding hole might have played a central role in connecting the other two holes during the filtration process.

\begin{figure}[ht]
\centering
\begin{subfigure}{0.49\textwidth}
    \centering
    \includegraphics[width=\linewidth]{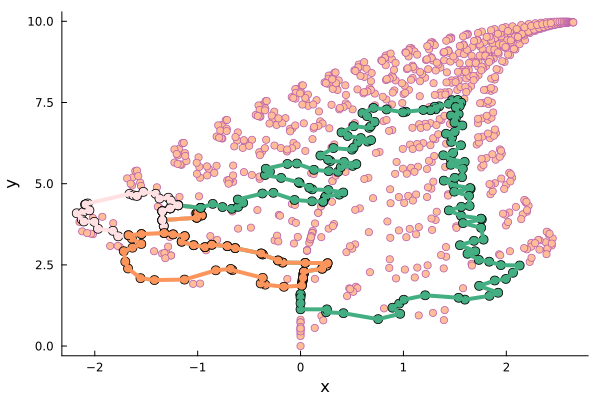}
    \caption{Three highest ranked holes}
    \label{fig:map_fern}
\end{subfigure}
\hfill
\begin{subfigure}{0.49\textwidth}
    \includegraphics[width=\linewidth]{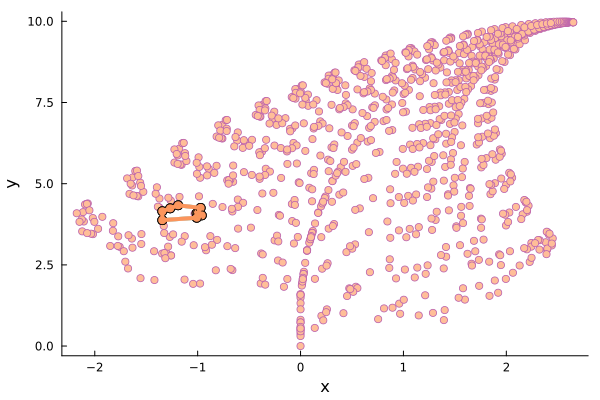}
    \caption{Earliest cycle across all orders of merge clusters}
    \label{fig:map1_fern}
\end{subfigure}
\caption{The earliest cycle across all orders of merge clusters of each of the three highest ranked holes}
\end{figure}

Similar to our previous analysis, we compared how centrality and persistence identify holes based on thresholds. Figure~\ref{fig:fern_number} illustrates the number of holes detected by each method at different threshold levels (represented as a percentage of the maximum persistence value). The plot reveals key differences in how these methods classify holes. At higher thresholds ($i > 0.6$), the centrality function $J_5$ identifies holes that meet the threshold. However, the persistence diagram identifies a larger number of holes, including those with lower persistence values ($i \leq 0.6$). This trend continues as we move towards lower thresholds. The persistence diagram consistently detects a higher number of holes (at least 50 holes at $i = 0.25$) compared to the centrality plot (which identifies at most 20 holes at $i = 0.25$). These observations reinforce our earlier finding that the centrality functions might be more selective in identifying features depending on a certain threshold.

\begin{figure}[ht]
\centering
\includegraphics[width=0.5\linewidth]{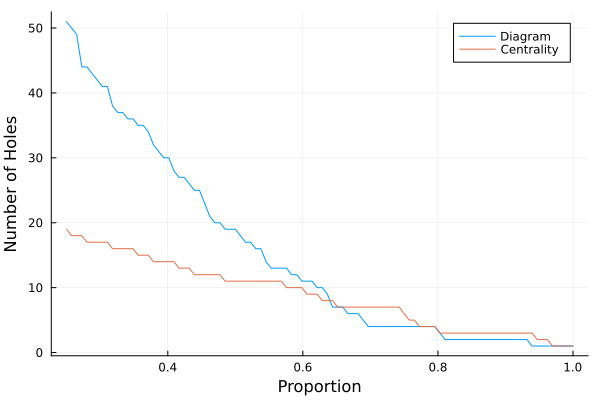}
\caption{Number of holes satisfying the threshold depending on the proportion of the maximum persistence and the maximum $J_5$ centrality values.}
\label{fig:fern_number}
\end{figure}

\section{Conclusion}

We introduced novel centrality measures that leverage both persistence and merge dynamics of homology classes. These measures aim to capture a more comprehensive picture of the topological structure within point cloud data compared to traditional summaries. The algorithm for computing the merge dynamics of homology classes is guided by the equivalence of merging and $q$-nearness between two classes. Similar to persistence barcodes, we have generated shape descriptors in the form of plots and heat maps and established stability by defining a pseudo-metric similar to bottleneck and landscape distances. Our initial investigation with self-similar point clouds has demonstrated agreement with existing TDA tools, while also revealing additional properties of important features.

Moving forward, we plan to assess the efficacy of these measures across diverse real-world point cloud datasets and within machine learning contexts. Furthermore, future research endeavors will focus on refining the centrality functions and investigating their mathematical properties in greater depth.

\section*{Acknowledgements}
John Rick Manzanares would like to express his gratitude to the Department of Science and Technology - Science Education Institute for supporting this work through the Accelerated Science and Technology Human Resource Development program. The same author would also like to sincerely thank Dr. Pawe\l \ D\l otko for his valuable suggestions and comments.

\printbibliography

@ARTICLE{fractal,
  title    = "Emergence of fractal geometries in the evolution of a metabolic
              enzyme",
  author   = "Sendker, Franziska L and Lo, Yat Kei and Heimerl, Thomas and
              Bohn, Stefan and Persson, Louise J and Mais, Christopher-Nils and
              Sadowska, Wiktoria and Paczia, Nicole and Nu{\ss}baum, Eva and
              del Carmen S{\'a}nchez Olmos, Mar{\'\i}a and Forchhammer, Karl
              and Schindler, Daniel and Erb, Tobias J and Benesch, Justin L P
              and Marklund, Erik G and Bange, Gert and Schuller, Jan M and
              Hochberg, Georg K A",
  journal  = "Nature",
  month    =  apr,
  year     =  2024
}

@article{spearman,
 ISSN = {00029556},
 URL = {http://www.jstor.org/stable/1422689},
 author = {C. Spearman},
 journal = {The American Journal of Psychology},
 number = {3/4},
 pages = {441--471},
 publisher = {University of Illinois Press},
 title = {The Proof and Measurement of Association between Two Things},
 urldate = {2024-04-18},
 volume = {100},
 year = {1987}
}

@book{bootstrap, 
place={Cambridge}, 
series={Cambridge Series in Statistical and Probabilistic Mathematics}, 
title={Bootstrap Methods and their Application}, 
publisher={Cambridge University Press}, 
author={Davison, A. C. and Hinkley, D. V.}, 
year={1997}, 
collection={Cambridge Series in Statistical and Probabilistic Mathematics}
}

@InProceedings{telnet,
author="He, Baozhu
and He, Zhen",
editor="Shen, Gang
and Huang, Xiong",
title="Centrality Measures in Telecommunication Network",
booktitle="Advanced Research on Electronic Commerce, Web Application, and Communication",
year="2011",
publisher="Springer Berlin Heidelberg",
address="Berlin, Heidelberg",
pages="337--343",
isbn="978-3-642-20370-1"
}

@article{socnet,
title = {Centrality in social networks conceptual clarification},
journal = {Social Networks},
volume = {1},
number = {3},
pages = {215-239},
year = {1978},
issn = {0378-8733},
doi = {https://doi.org/10.1016/0378-8733(78)90021-7},
url = {https://www.sciencedirect.com/science/article/pii/0378873378900217},
author = {Linton C. Freeman}
}

@Article{soccont,
author={Iacopini, Iacopo
and Petri, Giovanni
and Barrat, Alain
and Latora, Vito},
title={Simplicial models of social contagion},
journal={Nature Communications},
year={2019},
month={06},
day={06},
volume={10},
number={1},
pages={2485},
issn={2041-1723},
doi={10.1038/s41467-019-10431-6},
url={https://doi.org/10.1038/s41467-019-10431-6}
}

@Article{coauthor,
author={Carstens, C. J.
and Horadam, K. J.},
title={Persistent Homology of Collaboration Networks},
journal={Mathematical Problems in Engineering},
year={2013},
month={06},
day={04},
publisher={Hindawi Publishing Corporation},
volume={2013},
pages={815035},
issn={1024-123X},
doi={10.1155/2013/815035},
url={https://doi.org/10.1155/2013/815035}
}

@article{fractaldim,
title = {Fractal dimension estimation with persistent homology: A comparative study},
journal = {Communications in Nonlinear Science and Numerical Simulation},
volume = {84},
pages = {105163},
year = {2020},
issn = {1007-5704},
doi = {https://doi.org/10.1016/j.cnsns.2019.105163},
url = {https://www.sciencedirect.com/science/article/pii/S1007570419304824},
author = {Jonathan Jaquette and Benjamin Schweinhart}
}

@article{fractalpers,
url = {https://doi.org/10.1515/dema-2022-0015},
title = {Persistence landscapes of affine fractals},
author = {Michael J. Catanzaro and Lee Przybylski and Eric S. Weber},
pages = {163--192},
volume = {55},
number = {1},
journal = {Demonstratio Mathematica},
doi = {doi:10.1515/dema-2022-0015},
year = {2022},
lastchecked = {2024-04-17}
}

@Article{universal,
author={Bobrowski, Omer
and Skraba, Primoz},
title={A universal null-distribution for topological data analysis},
journal={Scientific Reports},
year={2023},
month={07},
day={28},
volume={13},
number={1},
pages={12274},
issn={2045-2322},
doi={10.1038/s41598-023-37842-2},
url={https://doi.org/10.1038/s41598-023-37842-2}
}

@Article{phom,
author={Zomorodian, Afra and Carlsson, Gunnar},
title={Computing Persistent Homology},
journal={Discrete {\&} Computational Geometry},
year={2005},
month={2},
day={01},
volume={33},
number={2},
pages={249-274},
issn={1432-0444},
doi={10.1007/s00454-004-1146-y},
url={https://doi.org/10.1007/s00454-004-1146-y}
}

@book{ctop,
	author = "Herbert Edelsbrunner and John Harer",
	title = "Computational topology : an introduction",
	publisher = "American Mathematical Society",
	address = "Providence, Rhode Island",
	year = "2010"
}

@book{ten,
	author = {Robert Ghrist},
	title = "Elementary applied topology",
	publisher = "CreateSpace",
	edition = "First",
	address = "California",
	year = "2014"
}

@article{landscape,
title = "A persistence landscapes toolbox for topological statistics",
journal = "Journal of Symbolic Computation",
volume = "78",
pages = "91 - 114",
year = "2017",
note = "Algorithms and Software for Computational Topology",
issn = "0747-7171",
doi = "https://doi.org/10.1016/j.jsc.2016.03.009",
url = "http://www.sciencedirect.com/science/article/pii/S0747717116300104",
author = "Peter Bubenik and Paweł Dłotko"
}

@Article{lumawig,
AUTHOR = {Ignacio, Paul Samuel and Bulauan, Jay-Anne and Uminsky, David},
TITLE = {Lumáwig: An Efficient Algorithm for Dimension Zero Bottleneck Distance Computation in Topological Data Analysis},
JOURNAL = {Algorithms},
VOLUME = {13},
YEAR = {2020},
NUMBER = {11},
ARTICLE-NUMBER = {291},
URL = {https://www.mdpi.com/1999-4893/13/11/291},
ISSN = {1999-4893},
DOI = {10.3390/a13110291}
}

@inproceedings{stability,
author = {Cohen-Steiner, David and Edelsbrunner, Herbert and Morozov, Dmitriy},
title = {Vines and Vineyards by Updating Persistence in Linear Time},
year = {2006},
isbn = {1595933409},
publisher = {Association for Computing Machinery},
address = {New York, NY, USA},
url = {https://doi.org/10.1145/1137856.1137877},
doi = {10.1145/1137856.1137877},
booktitle = {Proceedings of the Twenty-Second Annual Symposium on Computational Geometry},
pages = {119–126},
numpages = {8},
location = {Sedona, Arizona, USA},
series = {SCG '06}
}

@article{loopcentrality,
author={Giscard, Pierre-Louis and Wilson, Richard C.},
title={A centrality measure for cycles and subgraphs II},
journal={Applied Network Science},
year={2018},
month={6},
day={08},
volume={3},
number={1},
pages={9},
issn={2364-8228},
doi={10.1007/s41109-018-0064-5},
url={https://doi.org/10.1007/s41109-018-0064-5}
}

@article{curvature,
	doi = {10.1088/1361-6420/ab4ac0},
	url = {https://doi.org/10.1088/1361-6420/ab4ac0},
	year = 2020,
	month = {1},
	publisher = {{IOP} Publishing},
	volume = {36},
	number = {2},
	pages = {025008},
	author = {Peter Bubenik and Michael Hull and Dhruv Patel and Benjamin Whittle},
	title = {Persistent homology detects curvature},
	journal = {Inverse Problems}
}

@article{centrality,
title = "Centralities in simplicial complexes. Applications to protein interaction networks",
journal = "Journal of Theoretical Biology",
volume = "438",
pages = "46 - 60",
year = "2018",
issn = "0022-5193",
doi = "https://doi.org/10.1016/j.jtbi.2017.11.003",
url = "http://www.sciencedirect.com/science/article/pii/S0022519317305040",
author = "Ernesto Estrada and Grant J. Ross",
}

\end{document}